\documentclass[11pt,leqno,titlepage]{amsart}
\usepackage{amsmath}
\usepackage{amssymb}
\usepackage{enumerate}

\newtheorem{theorem}{Theorem}
\newtheorem{lemma}{Lemma}
\newtheorem{corollary}{Corollary}

\newtheorem{proposition}{Proposition}

\newcounter{claim_nb}[theorem]
\def\ba{\begin{array}}
\def\ea{\end{array}}
\def\beq{\begin{equation}}
\def\eeq{\end{equation}}
\def\bea{\begin{eqnarray}}
\def\eea{\end{eqnarray}}
\def\beann{\begin{eqnarray*}}
\def\eeann{\end{eqnarray*}}

\def\R{\mathbb{R}}
\def\S{\mathbb{S}}
\def\Z{\mathbb{Z}}

\def\cN{\mathcal{N}}
\def\cP{\mathcal{P}}

\begin{document}

\title[Approximately Minimizing Makespan and Flow-Time]{Worst-Case Performance Analysis
of some Approximation Algorithms for Minimizing Makespan
and Flow-Time}

\author{Peruvemba Sundaram Ravi, Levent Tun\c{c}el, Michael Huang}
\thanks{Peruvemba Sundaram Ravi ({\bf Corresponding Author}):
Operations and Decision Sciences,
School of Business and Economics, Wilfrid Laurier University,
Waterloo, Ontario N2L 3C5, Canada
(e-mail: pravi@wlu.ca, phone: (519)884-0710 ext.2674, fax: (519)884-0201)\\
Levent Tun\c{c}el: Department of Combinatorics and Optimization, Faculty
of Mathematics, University of Waterloo, Waterloo, Ontario N2L 3G1,
Canada (e-mail: ltuncel@math.uwaterloo.ca)\\
Michael Huang: (e-mail:Michael-Huang@u.northwestern.edu)}

\date{December 11, 2013}

\begin{abstract}
In 1976, Coffman and Sethi conjectured that a natural extension of LPT list scheduling to the bicriteria
scheduling problem of minimizing makespan over flowtime optimal schedules, called LD algorithm, has a simple
worst-case performance bound: $\frac{5m-2}{4m-1}$, where $m$ is the number of machines.
We study structure of potential minimal counterexamples to this conjecture
and prove that the conjecture holds for the cases (i) $n > 5m$, (ii) $m=2$, (iii) $m=3$, and (iv) $m\geq 4$, $n \leq 3m$,
where $n$ is the number of jobs.
We further conclude that to verify the conjecture, it suffices to analyze the following case:
for every $m \geq 4$, $n \in \{4m, 5m\}$.
\end{abstract}


\maketitle

\section{Introduction}

Various performance criteria may be used to schedule $n$ independent jobs on $m$ parallel identical machines. In applications where work-in-process
inventory is very highly valued (resulting in potentially huge work-in-process inventory costs), or in applications where the total time spent in the system
must be minimized, a fundamental objective function of choice would be
the minimization of total flow time (or equivalently mean flow time).  A schedule that minimizes mean flow time is termed a \emph{flowtime-optimal schedule}. The mean flow time for a set of jobs on a set of parallel identical
machines can be readily minimized by using the SPT (Shortest Processing Time) rule. The SPT rule generates a very large number (at least $\left({m!}\right)^{\lfloor n/m \rfloor}$) of flowtime-optimal schedules.
Another fundamental objective function of choice, as a secondary objective, would be the minimization of makespan. Minimizing the makespan has been shown to be an $\cN\cP$-hard problem. Various versions of this problem have been
studied by several researchers.

Graham (1966) examines the problem of makespan minimization for a set of jobs with a partial order (precedence constraints) on a set of parallel identical
machines. Graham (1969) as well as Coffman and Sethi (1976b) develop bounds for solutions obtained by the application of the LPT (Longest Processing Time) rule to the makespan minimization problem with no precedence constraints.The
bin-packing problem may be regarded as the dual problem to the makespan minimization problem. In the bin-packing problem, a set of items of---in general---unequal
sizes $\ell(T_i)$ must be packed into a set of $m$ bins, each of a given capacity $C$. The objective
of the bin-packing problem is to minimize the number of bins $m$ used for the packing.In the FFD (First Fit Decreasing) algorithm (Johnson (1973) and Johnson et al. (1974), the bins are assigned indices $1$ through $m$. The largest
remaining item is assigned to the lowest-indexed bin into which it can fit. The FFD solution is the smallest value of $m$ for which the algorithm is able to pack the entire set of items into m bins. Garey and Johnson (1981) and
Coffman, Garey and Johnson (1983) have proposed various alternatives to the
FFD algorithm for the bin-packing problem. Further, Coffman, Garey, and Johnson (1978) propose an algorithm for makespan minimization, the Multifit algorithm, that is based on the FFD algorithm for the bin-packing problem, and
obtain a bound for the performance of this algorithm. de la Vega and Lueker (1981) propose a linear-time approximation scheme for the bin packing
problem.  Friesen (1984),
Yue, Kellerer and Yu(1988), Yue (1990), and Cao(1995) also propose bounds for the performance of the Multifit algorithm. Dosa (2000 and 2001) proposes generalized versions of the LPT and Multifit methods. Chang and Hwang (1999)
extend the Multifit algorithm to a situation in which different processors have different starting times. Hochbaum and Shmoys (1987) propose a PTAS (Polynomial Time Approximation Scheme) for the makespan minimization problem for
parallel identical machines. Ho and Wong (1995) propose an $O(2^n)$ algorithm to find the optimal solution for a two-machine version of this problem.

The problem of selecting the schedule with the smallest makespan among the class of all flowtime-optimal schedules is known to be $\cN\cP$-hard (Bruno, Coffman and Sethi, 1974). We term this problem the
\emph{FM (Flowtime-Makespan) problem}. Coffman and Yannakakis (1984) study a more general version of the FM problem. They study the problem of permuting the elements within the columns of an $m$-by-$n$ matrix so as to minimize its
maximum row sum. Eck and Pinedo (1993) propose a new algorithm for the FM problem. Their algorithm, the LPT* algorithm, is a modified version of Graham's LPT algorithm. Their algorithm requires the construction of a new problem
instance by replacing every processing time in a rank by the difference between it and the smallest processing time in that rank. For the two-machine case, the authors obtain a worst-case bound of $28/27$ on the makespan ratio for
the LPT* algorithm. Gupta and Ho (2001) build on the procedure developed by Ho and Wong for makespan minimization to develop three algorithms - an algorithm to find the optimal solution and two heuristic procedures - for the
two-machine FM problem. Lin and Liao (2004) extend the procedures developed by Ho and Wong and by Gupta and Ho to construct a procedure to obtain the
optimal solution to the FM problem in $O\left(\left(m!\right)^{n/m}\right)$ time.

Conway, Maxwell and Miller (1967), in their seminal book, develop the notion of ranks for the FM problem. A schedule is flowtime-optimal if jobs are assigned in decreasing order of \textit{ranks}, with the jobs in rank $1$ being
assigned last. If $n$ is the number of jobs and $m$ is the number of machines, we may assume that $m$ divides $n$. (If it does not, we add $(\left\lceil n/m\right\rceil*m - n)$ dummy jobs with zero processing times.) If we assume
that the jobs are numbered in nonincreasing order of processing times, with job $1$ having the largest processing time, the set of jobs belonging to rank $r$ are the following: ${(r-1)m + 1, (r-1)m + 2,\cdots, (r - 1)m + m}$.
Coffman and Sethi (1976a) propose two approximation algorithms for the FM problem. In the LI algorithm, ranks are assigned in decreasing order, starting with the rank containing the $m$ jobs with the smallest processing times. In
the LD algorithm, ranks are assigned in increasing order, starting with the rank containing the $m$ jobs with the largest processing times. Jobs with the same rank are assigned largest-first onto distinct machines as they become
available after executing the previous ranks. In the LD algorithm, the sequence thus obtained must be reversed and all jobs in the last rank must be set to the same starting time of zero to ensure that the schedule is
flowtime-optimal. Coffman and Sethi show that the LI algorithm has a makespan ratio (ratio of the makespan to the optimal makespan) with a worst-case bound that is equal to $(5m-4)/(4m-3)$.

\bea
\label{Q1}
\begin{tabular}{|c|}\hline
\\
\parbox{4in}{Coffman and Sethi (1976) conjecture that the LD algorithm has a makespan ratio with a worst-case bound equal to
$$\frac{5m-2}{4m-1}. \,\,\,\,\,\,\,\,\,\,\,\,\,\,\,\,\,\,\,\,\,\,\,\,\,\,\,\,\,\,\,\,\,\,\,\,\,\,\,\,\,\,\,\,\,\,\,$$}\\
\\ \hline
\end{tabular}
\eea

\vspace{0.2cm}

The next example shows that the above conjectured ratio cannot be improved
for any $m \geq 2$.
For $m \geq 2$, let $n:=3m$ and consider
\[
p_j := \left\{\begin{array}{cl} 0 & \mbox{ for } j \in \{1,2,
\ldots,m-1\}\\
m & \mbox{ for } j =m\\
(j-1) & \mbox{ for } j \in \{m+1,m+2, \ldots, 2m\}\\
(j-2) & \mbox{ for } j \in \{2m+1, 2m+2, \ldots, 3m\}.
\end{array}
\right.
\]
It is easy to verify that the ratio of the objective value of an LD
schedule to the optimal objective value is $\frac{5m-2}{4m-1}$.

Our approach in obtaining a proof of Coffman and Sethi's conjectured bound is to identify properties of a hypothesized minimal counterexample to the conjecture.  We utilize three techniques:  (1) We show that there always exists a
minimal counterexample with integer processing times (even if we allow irrational data).
(2) Using integrality of the data and the other properties of the minimal counterexamples,
we construct one or more smaller problem instances by subtracting nonnegative integers from
the number of machines, the number of jobs, or the integer-valued processing times.
We then write constraints that capture the fact that these smaller problem instances must satisfy the conjecture.
We then prove that minimal counterexamples must have a very small number of ranks (say, three or four or five).  (3) The final stage of the proof is to verify that none of
these small instances (with a small number of ranks) can be counterexamples.  To establish this last bit, we introduce another
approach which is to treat processing times of jobs as
unknown variables and show how to set up the case analysis with a set of finitely many LP problems.  By solving these finitely many LP problems,
whose optimal objective function values can be certified (and easily verified) using primal and dual optimal solutions, we verify the conjecture for these small instances.
Our approach is based on general techniques and may be applicable to other combinatorial optimization problems.

The rest of this paper is organized as follows. Section 2 contains a description of the LI and LD algorithms. In Section 3, the notion of minimality that is used in this paper is defined. In Section 4, we discuss an LP-based
approach that treats problem parameters as variables. We use this approach to prove that the Coffman-Sethi conjecture holds for problem instances with two machines and for problem instances with three machines and three ranks
(this latter fact becomes a corollary of a theorem from in Section 7 and hence has two independent proofs). In Section 5, we state a result that implies that if the Coffman-Sethi conjecture is false, and even if processing times
are permitted to be irrational numbers, there must exist a counterexample to the conjecture with integer processing times. In Section 6, we demonstrate the effectiveness of our conjecture-proving procedure by applying it to a
simpler version of the LD algorithm. In Section 7, we derive several properties of a hypothesized minimal counterexample to the Coffman-Sethi conjecture. These include an upper bound on the number of ranks. Section 8 contains
concluding remarks and outlines possible extensions and directions for future research.

\section{Algorithms for Problem FM}

Let $p_j$ denote the processing time of job $j$, where the jobs are numbered in nonincreasing order of processing times. Thus $p_j \geq p_{j+1}$. The set of jobs belonging to rank $r$ are the following:
\[
{(r-1)m + 1, (r-1)m + 2,\cdots, (r - 1)m + m}.
\]

Any schedule in which all rank $(r+1)$ jobs are started before all rank $r$ jobs (where $r \in \{ 1,2,\cdots, (n/m) - 1\}$) is said to satisfy the \emph{rank restriction} or \emph{rank constraint}.

Any schedule in which all rank $(r+1)$ jobs are started before all rank $r$ jobs (where $r \in \{1,2,\cdots, (n/m) - 1\}$), there is no idle time between successive jobs assigned to the same machine, and all rank $n/m$ jobs start
at time $0$, is termed a \emph{flowtime-optimal schedule}. Let $\lambda_{r}$ and $\mu_{r}$ denote the largest and smallest processing times respectively in rank $r$. Note that
\bea \lambda_1 \geq \mu_1 \geq \lambda_2 \geq \mu_2 \geq \ldots \geq \lambda_{k-1} \geq \mu_{k-1} \geq \lambda_k \geq \mu_k \geq 0. \eea

The profile of a schedule after rank $r$ is defined as the sorted set of completion times on $m$ machines after rank $r$. If the jobs in $r$ ranks (out of a total of $k$) have been assigned to machines, and if the jobs in the
remaining ranks have not yet been assigned to machines, the profile after rank $r$ is termed the current profile. We let $a^{(r)} \in \Z^m: a_1^{(r)} \geq a_2^{(r)} \geq \cdots \geq a_m^{(r)}$ denote the current profile. Note that
$a_{i}(\ell)$ is the $i$th largest completion time after rank $\ell$.

\subsection{LI (Coffman and Sethi, 1976a)}
The LI algorithm starts with rank $k$, the rank containing the $m$ jobs with the
smallest processing times, and works its way through ranks $k, k-1, \ldots, 2, 1$. The schedule generated by the algorithm is thus a flowtime-optimal schedule.

The algorithm works as follows: Schedule the ranks in the following order:
\[
k, k-1, \ldots, 2,1.
\]
Let $a \in \Z^m$: \[ a_1 \geq a_2 \geq \cdots \geq a_m
\]
denote the current profile. Schedule the jobs in the next rank so that the largest processing time is matched with $a_m$, second largest with $a_{m-1}$, etc., and the smallest processing time
is matched with $a_1$. The algorithm moves from rank $k$ to rank $1$, with jobs within a rank being assigned based on the Longest Processing Time criterion.  Coffman and Sethi prove that
LI algorithm has performance ratio $\frac{5m-4}{4m-3}$ and that this ratio is tight for LI.

\subsection{LD (Coffman and Sethi, 1976a)}

The LD algorithm starts with rank $1$, the rank containing the $m$ jobs with the largest processing times, and works its way through ranks $1, 2, \ldots, k-1, k$. The sequence of jobs on each machine is then reversed to make it a
flowtime-optimal schedule.

The algorithm works as follows: Schedule the ranks in the following order:
\[
1, 2 , \ldots, k-1, k.
\]
Let $a \in \Z^m$: \[ a_1 \geq a_2 \geq \cdots \geq a_m
\]
denote the current profile. Schedule the jobs in the next rank so that the largest processing time is matched with $a_m$, second largest with $a_{m-1}$, etc., and the smallest processing time is matched with $a_1$. After all the
jobs are scheduled, reverse the schedule and left-justify it.

Note that LD algorithm is more naturally analogous to LPT than LI is.

\section{A definition of minimality}

If an instance of the FM problem with $m$ machines and $k$ ranks has fewer than $mk$ jobs, increasing the number of jobs to $mk$ by including up to $m-1$ jobs with zero processing time does not result in any change in the total
flowtime or the optimal makespan or the makespan of an LD schedule, though it may change the set of jobs allocated to each rank.  From this observation, it follows that, for the purpose of proving/disproving conjectures regarding
worst-case makespan ratios for algorithms for the problem FM, only problem instances with $mk$ jobs need to be considered.  Thus, from now on,
we assume the property
\[
\textup{(Property.1)} \,\,\,\,\,\,\,\,\,\, n=mk.
\]

We define the ordered set of processing times $P$ for a scheduling problem instance with $m$ machines and $k$ ranks to consist of elements equal to the processing times of these $mk$ jobs arranged in nonincreasing order. We let
$P(j)$ refer to the $j_{th}$ element of $P$. Thus, $P(j) \geq P(j+1)$ for $j \in \{1, \ldots, mk-1\}$.

For the problem of characterizing minimal counterexamples to the Coffman-Sethi conjecture, minimality will be defined based on the following criteria: the number of machines, the number of ranks,
and the set of processing times. A minimal counterexample is defined to be a counterexample with $k$ ranks, $m$ machines, and a set of processing times $P_{1}$ for which there does not exist another counterexample with one of the
following
(in a hierarchical order):\newline
(i) number of ranks fewer than $k$  \newline
(ii) number of ranks equal to $k$ and number of machines fewer than $m$ \newline
(iii) $k$ ranks, $m$ machines, and fewer jobs with nonzero processing times \newline
(iv) $k$ ranks, $m$ machines, the same number of jobs with nonzero processsing times,
and a worse approximation ratio \newline
(v) [only for integer data] $k$ ranks, $m$ machines, the same number of jobs with nonzero processing times,
the same approximation ratio, and a set of processing times $P_2$ satisfying\\
$\sum_{j=1}^n P_2(j) < \sum_{j=1}^n P_1(j)$ \newline

In the latter parts of the paper, we will restrict our attention purely to integer data.
There, we will use the minimality criterion (v).  Note that due to integrality of the
processing times, we will have a normalization of the data
(the smallest nonzero processing time is at least 1).

Throughout this paper, we will isolate many useful properties of a minimal counterexample. Note that
we may assume that in a minimal counterexample,
\[
\textup{(Property.2)} \,\,\,\,\,\,\,\,\, \mu_r=\lambda_{r+1}, \forall r \in \{1,2, \ldots, k-1\} \mbox{ and } \mu_k=0.
\]
If the above property fails, then either $\mu_k >0$ (in this case, we subtract $\mu_k$ from the processing time
of every job in rank $k$, the new instance is also a counterexample to the conjecture with the same number of machines,
ranks and fewer jobs with nonzero processing times, a contradiction to the minimality of the original instance)
or there exists $r \in \{1,2, \ldots, k-1\}$ such that $\mu_r > \lambda_{r+1}$ (in this case,  we subtract
$\left(\mu_r-\lambda_{r+1}\right)$ from the processing time of every job in rank $r$, the new instance
is also a counterexample to the conjecture with the same number of machines,
ranks, the number of jobs with nonzero processing times, and with a worse approximation
ratio, a contradiction to the minimality of the original instance).

We will show that, if the Coffman-Sethi conjecture is false,
then there must exist a counterexample with integer processing times. Further, if the conjecture is false,
then there must exist a counterexample with integer processing times and a rectangular optimal schedule.
Therefore, we define four types of problem instances, counterexamples and minimal counterexamples to the Coffman-Sethi
conjecture:
 \begin{itemize}
 \item
 A problem instance of \emph{Type A} is one that is not required to have either integer processing times or a rectangular optimal schedule.
 \item
 A problem instance of \emph{Type I} is one with integer processing times.
 \item
 A problem instance of \emph{Type R} is one with a rectangular optimal schedule.
 \item
 A problem instance of \emph{Type IR} is one with integer processing times and a rectangular optimal schedule.
 \end{itemize}

 A counterexample of a particular type (A, I, R, or IR) is a problem instance of that type that violates the Coffman-Sethi conjecture. A \emph{minimal counterexample} of a particular type is a counterexample of that type for which
 there does not exist a smaller counterexample (based on the notion of minimality defined in this section) of the same type. In Section 5, we will show that if the Coffman-Sethi conjecture is false, then there exists a minimal
 counterexample of Type I, as well as of Type IR.

\section{An LP-based approach that treats problem parameters as variables}

We present an approach which may be utilized to prove or disprove any conjecture about a bound on the ratio of the solution generated by an approximation algorithm to the optimal solution for a maximization problem. In this
approach, the parameters of the problem are treated as variables. The problem of maximizing the ratio is formulated as a linear program. If bounds can be obtained on the size of the problem, only a finite number of linear programs
need to be examined. If none of the solutions to the linear programs violates the conjecture, the conjecture is clearly valid.
Indeed, the optimal objective value of each linear programming problem can be verified by checking a pair of optimal solutions
to the primal and the dual problem at hand.

In this section, this approach is used to show that the Coffman-Sethi conjecture is valid for the $m = 2$ and the $m =3, k = 3$ cases, where $m$ denotes the number of machines and $k$ denotes the number of ranks. In subsequent
sections, bounds are obtained on the number of ranks. In principle, this LP-based approach can be used to obtain a computer-based proof of the conjecture or a minimal counterexample for any finite number of machines.
				
\begin{lemma}
\label{lem:4.1}
If the Coffman-Sethi conjecture is false for $m=2$ then there exists a minimal counterexample of Type A with two machines and three ranks.
\end{lemma}

\begin{proof}
Suppose that the Coffman-Sethi conjecture is false.  Take a minimal counterexample to the conjecture.
Then, by (Property.2), in rank $k$, one of the two machines has a processing time of $\lambda_{k}$ and the other machine has a processing time of $0$. Clearly, the makespan is equal to the completion time after rank $k$ on the
machine with a processing time of $\lambda_{k}$. (If this is not the case, the last rank could be deleted to obtain a problem instance with the same or larger makespan ratio, contradicting the minimality of the original instance.)
It follows that both completion times on the two machines after rank $(k-1)$ in the LD schedule are at least $\left(t_{LD} - \lambda_{k}\right)$.   This implies $t^{*} \geq t_{LD} - \lambda_{k} + \lambda_{k}/2$. In a counterexample,
$t_{LD}/t^{*} > 8/7$.  Hence,
\bea
\label{eq.lem:4.1}
t^{*} & < & 7\lambda_{k}/2.
\eea
Clearly, $\mu_{\ell} \geq \lambda_{k}$ for $\ell \in \{1,2, \ldots, k-1\}$.  Thus, $t^{*} \geq k\lambda_k$ which together with \eqref{eq.lem:4.1}
yields $k \leq 3$.  Note that for $k = 1$ and $k = 2$, the LD schedule is optimal. Therefore, the claim follows.
\end{proof}

The following approach treats the data of problem FM (job processing times) as variables. For a given value of the optimal makespan, the problem of determining the values of the processing times that result in the LD makespan being
maximized is set up as a set of linear programs. Each possible relationship between the processing times (subject to the rank restriction) results in a different linear program. The solution to each linear program is checked to
determine if it violates the Coffman-Sethi conjecture.  Note that a major advantage of LP formulations is that the optimal objective
values can be verified independent of the original computations, by a much simpler computational step (to check feasibility of primal and dual solutions and then
a comparison of their objective values).

\begin{proposition}
\label{lem:4.2}
The Coffman-Sethi conjecture holds for $m=2$. \end{proposition}

\begin{proof}
By Lemma \ref{lem:4.1}, we only need to consider the $k=3$ case.  For a contradiction,
suppose the conjecture is false.  Then, there exists a minimal counterexample to the
conjecture.  Moreover, using (Property.2), we
may assume that a minimal counterexample has the processing times:
$\lambda_1, \lambda_2, \lambda_2, \lambda_3, \lambda_3, 0$.  Then,
it suffices to consider only two LD schedules:
\begin{itemize}
\item
LD schedule 1: Jobs with processing times $\lambda_{1}, \lambda_{3}, 0$ on machine $1$, jobs with processing times $\lambda_{2}, \lambda_{2}, \lambda_{3}$ on machine $2$.
\item
LD schedule 2: Jobs with processing times $\lambda_{1}, \lambda_{3}, \lambda_{3}$ on machine $1$, jobs with processing times $\lambda_{2}, \lambda_{2}, 0$ on machine $2$.
\end{itemize}

For a makespan ratio $> 1$, the second and third ranks must not be the same in the LD schedule and the optimal schedule. There is only one possible optimal schedule: Jobs with processing times $\lambda_{1}, \lambda_{2}, 0$ on
machine $1$, jobs with processing times $\lambda_{2}, \lambda_{3}, \lambda_{3}$ on machine $2$.  In each of the following cases, we set the optimal makespan equal to $1$. The makespan ratio is then equal to the LD makespan. We seek
to maximize the LD makespan in each case.  There are four possible values for the LD makespan, resulting in the following four cases.

\underline{Case 1:} $t_{LD} = \lambda_{1} + \lambda_{3}$.\\
This will be true only if $\lambda_{1} \geq 2\lambda_{2}.$ So,
$\lambda_{1} \geq 2\lambda_{3}$, and we deduce
$t_{S^{*}} = \lambda_{1} + \lambda_{2}.$
This is clearly not possible.

\underline{Case 2:} $t_{LD} = 2\lambda_{2} + \lambda_{3}$.\\
This will be true only if $\lambda_{1} \leq 2\lambda_{2}$
and $\lambda_{1} + \lambda_{3} \geq 2\lambda_{2}$.

\underline{Case 2A:} $\lambda_{1} \leq 2\lambda_{3}$.\\
So, $t_{S^{*}} = \lambda_{2} + 2\lambda_{3}.$  Consider the LP problem
\[ \max \left\{2\lambda_{2} + \lambda_{3}:
\,\, \lambda_{1} \leq 2\lambda_{3}, \,\, \lambda_{2} + 2\lambda_{3} = 1,
\,\, 2\lambda_{2} \leq \lambda_{1} + \lambda_{3} \right\}.
\]			
This LP can be simplified as follows:
\[ \max \left\{2\lambda_{2} + \lambda_{3}:
\,\, 2\lambda_{2} \leq 3\lambda_{3}, \,\, \lambda_{2} + 2\lambda_{3} = 1 \right\}.
\]		
A solution is $\lambda_{2} = 3/7, \lambda_{3} = 2/7$, with objective function
value $8/7$.

\underline{Case 2B:} $\lambda_{1} \geq 2\lambda_{3}.$\\
Thus, $t_{S^{*}} = \lambda_{1} + \lambda_{2}.$
Consider the LP problem
\[
\max \left\{2\lambda_{2} + \lambda_{3}: \,\,
\lambda_{1} + \lambda_{2} = 1, \,\, \lambda_{1} \geq 2\lambda_{3}, \,\,
\lambda_{1} \leq 2\lambda_{2}, \,\, 2\lambda_{2} \leq \lambda_{1} + \lambda_{3}
\right\}.
\]			
This LP can be simplified as follows:
\[
\max \left\{2\lambda_{2} + \lambda_{3}: \,\,
\lambda_{2} \geq 1/3, \,\, \lambda_{2} + 2\lambda_{3} \leq 1, \,\,
3\lambda_{2} \leq 1 + \lambda_{3} \right\}.
\]
A solution is $\lambda_{2} = 3/7, \lambda_{3} = 2/7$, with the objective function value $8/7$.

\underline{Case 3:} $t_{LD} = \lambda_{1} + 2\lambda_{3}$.\\
This will be true only if $\lambda_{1} + \lambda_{3} \leq 2\lambda_{2}$
and $\lambda_{1} + 2\lambda_{3} \geq 2\lambda_{2}$.

\underline{Case 3A:} $\lambda_{1} \leq 2\lambda_{3}$.\\
In this case, we have $t_{S^{*}} = \lambda_{2} + 2\lambda_{3}$.  Consider the LP
problem:
\[
\max \left\{\lambda_{1} + 2\lambda_{3}: \,\,
\lambda_{1} + \lambda_{3} \leq 2\lambda_{2}, \,\,
\lambda_{1} + 2\lambda_{3} \geq 2\lambda_{2}, \,\,
\lambda_{2} + 2\lambda_{3} = 1, \,\,
2\lambda_{1} \leq 2\lambda_{3} \right\}.
\]
The constraints of the above LP problem
can be replaced by the following equivalent set of constraints:
\[\lambda_{1} + 5\lambda_{3} \leq 2, \,\, \lambda_{1} + 6\lambda_{3} \geq 2,
\,\, 2\lambda_{1} \leq 2\lambda_{3}.
\]			
A solution is $\lambda_{1} = 4/7, \lambda_{2} = 3/7, \lambda_{3} = 2/7$,
with objective function value $8/7$.

\underline{Case 3B:} $\lambda_{1} \geq 2\lambda_{3}$.\\
Thus, $t_{S^{*}} = \lambda_{1} + \lambda_{2}$.  Consider the LP
problem
\[
\max \left\{\lambda_{1} + 2\lambda_{3}: \,\,
\lambda_{1} + \lambda_{2} = 1, \,\,
\lambda_{1} \geq 2\lambda_{3}, \,\,
\lambda_{1} + \lambda_{3} \leq 2\lambda_{2}, \,\,
\lambda_{1} + 2\lambda_{3} \geq 2\lambda_{2} \right\}.
\]		
This LP problem can be simplified as follows:
\[
\max \left\{\lambda_{1} + 2\lambda_{3}: \,\,
\lambda_{1} \geq 2\lambda_{3}, \,\,
3\lambda_{1} + \lambda_{3} \leq 2, \,\,
3\lambda_{1} + 2\lambda_{3} \geq 2 \right\}.
\]		
A solution is $\lambda_{1} = 4/7, \lambda_{2} = 3/7, \lambda_{3} = 2/7$,
with the objective function value $8/7$. 				

\underline{Case 4:} $t_{LD} = 2\lambda_{2}$.\\
This will be true only if $2\lambda_{2} \geq \lambda_{1} + 2\lambda_{3}$.

\underline{Case 4A:} $\lambda_{1} \leq 2\lambda_{3}$.\\
Hence, $t_{S^{*}} = \lambda_{2} + 2\lambda_{3}$.
Consider the LP problem
\[
\max \left\{2\lambda_{2}: \,\,
\lambda_{1} \leq 2\lambda_{3}, \,\,
\lambda_{2} + 2\lambda_{3} = 1, \,\,
2\lambda_{2} \geq \lambda_{1} + 2\lambda_{3} \right\}.
\]		
The constraints of the above LP problem
can be replaced by the following equivalent set of constraints:
\[
\lambda_{1} + \lambda_{2} \leq 1, \,\,
-2\lambda_{1} + 5\lambda_{2} \geq 1.
\]		
A solution is $\lambda_{1} = 1/2, \lambda_{2} = 1/2, \lambda_{3} = 1/4$,
with the objective function value $1$.
 				
\underline{Case 4B:} $\lambda_{1} \geq 2\lambda_{3}$.\\
Then, $t_{S^{*}} = \lambda_{1} + \lambda_{2}$.  Consider the LP problem
\[
\max \left\{2\lambda_{2}: \,\,
\lambda_{1} + \lambda_{2} = 1, \,\,
\lambda_{1} \geq 2\lambda_{3}, \,\,
\lambda_{1} + 2\lambda_{3} \leq 2\lambda_{2}
\right\}.
\]					
This LP problem can be simplified as follows:
\[
\max \left\{2\lambda_{2}: \,\,
\lambda_{2} \geq \lambda_{3}, \,\,
2\lambda_{3} + \lambda_{2} \leq 1, \,\,
-2\lambda_{3} + 3\lambda_{2} \geq 1
\right\}.
\]		
A solution is $\lambda_{1} = 1/2, \lambda_{2} = 1/2, \lambda_{3} = 1/4$,
with the objective function value $1$.

We conclude that, in all cases, the ratio of the makespan of every LD schedule is at most $8/7$ times the optimum
makespan.  This is a contradiction to the existence of a counterexample.
Thus, the Coffman-Sethi conjecture holds for $m=2$, and $k=3$.  Therefore, by Lemma \ref{lem:4.1},
the Coffman-Sethi conjecture holds for $m=2$.\end{proof} 	

It follows from Graham's original analysis of LPT list scheduling that $n \leq 2m$ case is clear.

\begin{proposition}
\label{prop:2}
The Coffman-Sethi conjecture holds for every $m$ and $n$ such that $n \leq 2m$ (i.e., $k \leq 2$). \end{proposition}

Using the same technique as in the proof of Proposition \ref{lem:4.2} we can prove the
conjecture for $m=3$ and $k=3$.

\begin{proposition}
\label{lem:4.3}
The Coffman-Sethi conjecture holds for the $m=3$, $k=3$ case. \end{proposition}

A proof of the above proposition is provided in the appendix.

\section{Sufficient conditions for the existence of a minimal counterexample of Type I}

Most scheduling problems are considered in the context of the Turing machine model of computation and, as a result, the data are assumed to be drawn from the rationals. In our problem FM, this would mean that $p_j \in \mathbb{Q},
\forall j \in \{1,2,\ldots,n\}$. Below, we prove that
a counterexample to a conjectured makespan ratio with irrational processing times could exist only if there existed a counterexample with rational processing times.

Let $p \in \R^{n}_+$ denote the set of processing times for an instance $E$ of the problem FM. Let $p = (p_1, p_2, \ldots, p_n)$. Let $t^*_p$ denote the optimal value for instance $E$.

An instance of the FM problem is defined by an integer $m$ denoting the number of machines and a set of processing times $(p_1, p_2,\ldots, p_n)$. Therefore, the domain of FM is a  subset of $\Z_+ \times \R^n_+$. Let $FM(m,n)$
denote the FM problem with a fixed value of $m$ and a fixed value of $n$. An instance of $FM(m,n)$ is defined by $(m,n; p_1, p_2, \ldots, p_n)$. It follows that the domain of $FM(m,n)$ is a subset of $\Z^2_+ \times \R^n_+$.
We denote the domain by $\emph{G}$ and we describe below more general results.
For the details and the proofs of the following two results, see Ravi (2010).

\begin{proposition} \label{TYPE_I} For every algorithm $ALG$ for problem FM(m,n) that produces a feasible solution $\S_{ALG}$ whose makespan $t_{{ALG}}: \emph{G} \rightarrow \R_+$ is a continuous function at every point in
$\emph{G}$ and for a conjecture which states that the makespan is no greater than $\emph{f}(m). t^*$, where $\emph{f}: \Z_{+} \rightarrow \R_{+}$, the following must hold: If there exists a counterexample of Type A to a conjectured
$t_{ALG}/t^*$ ratio, then there exists a counterexample of Type I.
\end{proposition}

\begin{corollary} \label{LD_TYPE_I} For the FM problem and the LD algorithm, the following must hold: If there exists a counterexample E to a conjectured $t_{LD}/t^*$ ratio, then there exists a counterexample EI with integer
processing times.\end{corollary}

From this corollary, it follows that, if the Coffman-Sethi conjecture is false, then there exists a counterexample to the conjecture of Type I. If there exist one or more counterexamples of Type I, there must exist a minimal
counterexample of Type I.

\section{Best possible bound for the makespan ratio of $LD_0$ schedules}

Coffman and Sethi (1976) defined a simpler version (called $LI_0$) of their algorithm and analyzed its performance.
Let us define the corresponding analogue of $LI_0$ as follows.  An \emph{$LD_0$ schedule} is defined to be a flowtime-optimal schedule
in which the second rank is assigned largest-first. An $LD_{0}$ schedule is constructed using the following four-step process:
\begin{enumerate}[(i)]
\item
Jobs in the first rank are assigned arbitrarily to machines.
\item
Jobs in the second rank are assigned largest-first. Thus the jobs are assigned in nonincreasing order of processing times to the earliest available machine.
\item
Jobs in each remaining rank are assigned arbitrarily to machines.
\item
The jobs assigned to each machine are reversed to make the schedule flowtime-optimal.
\end{enumerate}
A \emph{worst-case ${LD}_0$ schedule} for a given set of tasks is an $LD_{0}$ schedule with the largest length among all $LD_{0}$ schedules for that set of tasks. A problem instance with a \emph{worst-case $t_{{LD}_0}/t^*$}
ratio is one with the largest $t_{{LD}_0}/t^*$ ratio among all flowtime-optimal schedules.
Let us define the $LD_{0_{worst}}$ schedule to be an $LD_{0}$ schedule with the largest makespan ratio among all $LD_{0}$ schedules for a given problem instance. Clearly, an $LD_{0_{worst}}$ schedule is constructed by assigning the
second rank largest-first and, if machine $i'$ has the largest completion time after rank $2$, assigning a job with processing time ${\lambda}_r$ to machine $i'$ for $r \in \{3, \ldots, k\}$.
Let $t_{LD_{0_{worst}}}$ denote the makespan of the $LD_{0_{worst}}$ schedule.

\begin{corollary} For the FM problem and the $LD_{0}$ algorithm, the following must hold: If there exists a counterexample E to a conjectured $t_{{LD}_{0}}/t^*$ ratio, then there exists a counterexample EI with integer processing
times.\end{corollary}

\begin{proof} The proof of this corollary is essentially identical to the proof of the previous corollary.\end{proof}

The following lemmas are useful because they make it possible to subsequently restrict our attention to rectangular schedules for an examination of the worst-case makespan ratio for $LD_{0}$ schedules.

\begin{proposition} \label{NONDECREASING_LD_0} Increasing the processing times of one or more tasks in the first rank of a $LD_{0_{worst}}$ schedule, while leaving the remaining processing times unchanged, will result in a
$LD_{0_{worst}}$ schedule with the same or larger makespan.\end{proposition}
\begin{proof} Let $a(2)$ denote the profile of the $LD_{0_{worst}}$ schedule after rank $2$. $a_{1}(2)$ is the length of the set of tasks upon completion of the second rank in an $LD_{0_{worst}}$ schedule. An $LD_{0_{worst}}$
schedule is obtained by assigning the tasks with the largest processing times in ranks $3$ through $k$, where $k$ is the last rank, to a machine with completion time $a_{1}(2)$ after rank $2$. So,
\[
t_{{LD}_{0_{worst}}} = a_{1}(2) +  \sum_{i=3}^{k}\lambda_{j}, \mbox{ with } a_{1}(2) =\max_{i\in\{1,\ldots,m\}}\left\{\tau_{i,1} + \tau_{m-i+1,2}\right\},
\]
where $\tau_{i,j}$ refers to the $i^{th}$ largest processing time in rank $j$. Leaving all $\tau_{i,2}$ values unchanged and increasing one or more $\tau_{i,1}$ values can only increase the value of $a_{1}(2)$ and thus of
$t_{{LD}_{0_{worst}}}$. \end{proof}

\begin{lemma} For every problem instance $P_{A}$ for the FM problem, there exists a problem instance $P_{B}$ with a rectangular optimal schedule and a $t_{{LD}_{0}}/t^*$ ratio that is at least the makespan ratio for
$P_{A}$.\end{lemma}

\begin{proof} Consider an optimal schedule for $P_{A}$. Construct a new problem instance $P_B$ as follows: For every job in the first rank that is performed on a machine for which the completion time after rank $k$ in the optimal
schedule is less than the makespan of the optimal schedule, increase the processing time so that the completion time after rank $k$ in the optimal schedule is equal to the makespan. Clearly, the optimal schedule for the new problem
instance is a rectangular schedule with a makespan equal to the makespan of $P_{A}$. Also, the new problem instance continues to be a flowtime-optimal problem instance because the increase in processing times does not lead to a
violation of the rank restrictions. From Proposition \ref{NONDECREASING_LD_0}, the makespan of the $LD_{0_{worst}}$ schedule for the new problem instance is at least the makespan of the $LD_{0_{worst}}$ schedule for the original
problem instance. Therefore, $P_B$ has a $t_{{LD}_{0_{worst}}}/t^*$ ratio that is at least the makespan ratio for $P_{A}$. \end{proof}

\begin{corollary} \label{RECTANGULAR_LD_0} There always exists a problem instance with a worst-case $t_{{LD}_{0}}/t^*$ ratio and a rectangular optimal schedule.
\end{corollary}

\begin{lemma} \label{CONDITIONS} There exists a problem instance with a rectangular optimal schedule and a worst-case $t_{{LD}_{0}}/t^*$ ratio in which the following hold:\\
(i) All the tasks in the second rank have a processing time equal to $\lambda_{3}$.\\
(ii) There exist one or more tasks in the first rank with a processing time equal to $\lambda_{3}$.\\
Thus, $\mu_1 = \lambda_2 = \mu_2 = \lambda_3.$
\end{lemma}
\begin{proof} By Corollary \ref{RECTANGULAR_LD_0}, there always exists a problem instance with a rectangular optimal schedule and a worst-case $t_{{LD}_{0_{worst}}}/t^*$ ratio. Let $\S^*$ denote this rectangular optimal schedule.
Let the machines be labelled based on the completion times after rank $2$ in the optimal schedule $\S^*$. Thus, the machine with the largest completion time after rank 2 in $\S^*$ is labelled machine $1$, the machine with the
second-largest completion time after rank $2$ in $\S^*$ is labelled machine $2$, and the machine with the $q^{th}$ largest processing time (for $q \in \{1,2,...,m\}$) is labelled machine $q$. Recall that $a_{i}(2)$ denotes the
$i^{th}$ largest element of the $\S^*$ profile after rank $2$. This labelling ensures that machine $i$ (for $i \in \{1,2,...,m\}$) has completion time $a_{i}(2)$ after rank $2$.

For machine $i \in \{1,2,\ldots, m\}:$ Set the processing time of the task in the second rank equal to $\lambda_{3}$ and the processing time of the task in the first rank equal to $a_{i}(2) - \lambda_{3}$. This change will either
increase or leave unchanged the processing time of each job in rank $1$. Therefore, it will not violate the rank restriction. Note that $a_{i}(2)$ remains unchanged and therefore the value of $t^*$ remains unchanged.

For the $LD_{0_{worst}}$ schedule, the set of changes in the task set described above is equivalent to:
(i) Reassigning the tasks in rank $2$ so that each task in rank $2$ is placed after the same task in rank $1$ as it was in the $\S^*$ schedule.
(ii) For machine $i \in \{1,2,\ldots,m \} $ (labelled as indicated above): Set the processing time of the task in the second rank equal to $\lambda_{3}$ and the processing time of the task in the first rank equal to $a_{i}(2) -
\lambda_{3}$.

Note that the resulting schedule is an $LD_{0_{worst}}$ schedule. An $LD_{0_{worst}}$ schedule is optimal for a two-rank system. The rearrangement of the tasks in step (i) can therefore only increase the value of $a_{i}(2)$.
\[
t_{LD_{0_{worst}}} = a_{i}(2) + \sum_{j=3}^{k}\lambda_{j}.
\]
Therefore, the value of $t_{{LD}_{0_{worst}}}$ can only increase as a result of the rearrangement in step (i). The subsequent changes in step (ii) do not cause any further change in $a_{i}(2)$ and therefore do not cause any further
change in $t_{{LD}_{0_{worst}}}$.

Thus, the value of $t_{{LD}_{0_{worst}}}$ can only increase as a result of steps (i) and (ii) while the value of $t^*$ remains unchanged. Therefore, the ratio $t_{{LD}_{0_{worst}}}/t^*$ can only increase.

After completing steps (i) and (ii),
if $a_{m}(1) > \lambda_{3}$, for $i \in \{1,2,\cdots\ ,m\},$ set $a_{i}(1)$ equal to $a_{i}(1) - ( a_{m}(1) - \lambda_{3})$.
It follows that this will reduce $t_{{LD}_{0_{worst}}}$ and $t^*$ by the same amount and will therefore increase the ratio $t_{{LD}_{0_{worst}}}/t^*$.
This will result in a value of $\mu_{1}$ equal to $\lambda_{3}$. \end{proof}

Coffman and Sethi (1976) proved that the ratio of the makespan of any schedule in which the rank containing the largest processing times was assigned largest-first to the makespan of the optimal schedule could not exceed ($4m -
3$)/($3m - 2$), and this bound could be achieved for $m = 2$ and for $m = 3$. The following theorem provides a similar bound for $LD_{0_{worst}}$ schedules. It shows that the worst-case ratio for $LD_{0_{worst}}$ schedules cannot
exceed $4/3$. While Coffman and Sethi suggest that their ($4m - 3$)/($3m - 2$) bound is unlikely to be achieved for $m > 3$, this bound can be achieved for all $m$. The proof below uses an approach that is different from the
approach used by Coffman and Sethi.

\begin{theorem} $1 \leq t_{{LD}_{0}}/t^* \leq 4/3$, and the upper bound is achieved for all $m \geq 3$.
\end{theorem}
\begin{proof}
Consider a problem instance $PI$ that satisfies the conditions established in Lemma \ref{CONDITIONS}. All the tasks in the second rank of a flowtime-optimal schedule for this problem instance have a processing time equal to
$\lambda_3$. Let $t$ denote the makespan of $\S$, the $LD_{0_{worst}}$ schedule for this problem instance. Let $t^*$ denote the optimal makespan for this problem instance.

Now, construct a new problem instance $PI'$ by removing the jobs in the second rank from $PI$. Consider the schedules obtained by removing the jobs in the second rank from $\S$ and $\S^*$.
Coffman and Sethi (1976) show that, for any problem instance with a flowtime-optimal schedule, the makespan ratio is less than or equal to $3/2$. It follows that
$(t - \lambda_3)/(t^* - \lambda_3)  \leq 3/2.$  The last relation is equivalent to
$t/t^* \leq 3/2 - \frac{\lambda_{3}}{2t^*}.$
Note that $\S^*$ is rectangular. Therefore, the machine with processing time $\mu_{1}$ in rank $1$ has a completion time at the end of rank $k$ that is equal to the length of $\S^*$.
An upper bound on $t^*$ can be obtained by adding the largest processing time in ranks $2$ through $k$ to $\mu_{1}$. Therefore,
\[
t^* \leq  \mu_{1} + \sum_{j=2}^{k}\lambda_{j} \leq 3\lambda_{3} + \sum_{j=4}^{k}\lambda_{j}.
\]

By Lemma \ref{TYPE_I} and Proposition \ref{TYPE_I}, it follows that, if there exists a counterexample to the $4/3$ bound, there exists a counterexample in which all processing times are integers. Also note that, if a counterexample
to the $4/3$ bound exists, a problem instance with a worst-case $t_{{LD}_{0}}/t^*$ ratio would be a counterexample. Let us assume that the $4/3$ conjecture is false. For a problem instance with a worst-case $t_{{LD}_{0}}/t^*$ ratio
and integer processing times, we have
\[
\lambda_{i} \leq \lambda_{3} - i + 3 \mbox{ for } i \in \{4,5, \ldots, k\}.
\]
Thus, $\sum_{j=4}^{k}\lambda_{j} \leq (k-3)\lambda_{3} - \frac{1}{2}\left(k^2-5k+6\right)$ and we conclude
\[
t^* \leq k\lambda_3 - \frac{1}{2}\left(k^2-5k+6\right).
\]
For $k \geq 3$, and assuming integer processing times, $\lambda_3$ must be greater than or equal to $k - 2$. For $k \in [3, \lambda_3+2]$,
the right-hand side is monotone increasing in $k$. For $k = 3$, the right-hand side is equal to $3\lambda_3$. Therefore, for $k \geq 3$, the right-hand side is greater than or equal to $3\lambda_3$. It follows that
\[
t_{{LD}_{0_{worst}}}/t^* \leq 3/2 - (1/2)(1/3) =  4/3.
\]
However, this contradicts the assumption that there exists a counterexample to the $4/3$ bound. It follows that no such counterexample exists and the bound is valid.

For any value of $m \geq 2$, the $t_{{LD}_{0_{worst}}}/t^*$ ratio equals the upper bound for the following three-rank system:
\beann
p_i = 2 &\mbox{ for }&  i \in \{1,2, \ldots, m-1\},\\				
p_i = 1 &\mbox{ for }&  i \in \{m, m+1, \ldots, 2m+1\},\\				
p_i = 0 &\mbox{ for }& i \in \{2m+2, 2m+3, \ldots, 3m\}.\\
\eeann	
This completes the proof of the theorem. \end{proof}

\section{Properties of a hypothesised minimal counterexample to the Coffman-Sethi conjecture}
The results obtained above for the simple ${LD}_0$ algorithm provided insight into the approaches that could be used to prove bounds for other algorithms for problem FM. This section will focus on the LD algorithm proposed by
Coffman and Sethi (1976a).

\begin{lemma} \label{NONDECREASING_PROFILE} An increase in one or more processing times of jobs in rank $r$ for $r \in\{1,2, \ldots,k -1\}$ (with no change in the remaining processing times, and subject to the rank constraint) does
not result in a reduction in any element of the profile $b(\ell)$ of an LD schedule after rank $\ell \in \{r, r+1, \ldots, k\}$.
\end{lemma}
\begin{proof} We proceed by induction on $\ell$.
Let us assume that the lemma holds for $\ell'$ ranks, where $\ell' \in \left\{r,r+1,\ldots,s\right\}$.
Let $\tau_{i,h}$ refer to the $i^{th}$ largest processing time in rank $h$.
The induction hypothesis states that an increase in processing times in rank $r$ does not cause a reduction in $b_{m-i+1}(\ell')$ for $i \in \{1,2,\ldots, m\}.$ Note that the increase in processing times in rank $r$ leaves
$\tau_{i,\ell' + 1}$ unchanged for $i \in \{1,2, \ldots, m\}.$
The profile $b(\ell' + 1)$ after rank $\ell' + 1$ consists of the following $m$ elements:
\[
\left[b_{m-i+1}(\ell') + \tau_{i,\ell' + 1}\right] \mbox{ for } i \in \{1,2,\ldots, m\}.
\]
It follows that none of the elements of the profile $b(\ell' + 1)$ gets reduced as a result of the increase in processing times in rank $r$.
Thus, the theorem holds for rank $\ell' + 1$, for $\ell' \in \left\{r,r+1,\ldots,s\right\}$.
The base case follows from the fact that the result holds trivially for $\ell' = r$.\end{proof}

One important fact follows from the above lemma. An optimal flowtime-optimal schedule that is not rectangular can be made rectangular by increasing the lengths of all tasks in the first rank that are performed on machines that have
a completion time after the last rank that is strictly less than the makespan. This result (originally observed by Coffman and Sethi) is stated and proved below for the sake of completeness.

\begin{lemma}\label{lem:CS1} (Coffman and Sethi, 1976a)
There always exists a problem instance with a worst-case $t_{LD}/t^*$ ratio and with
a rectangular optimal schedule.
\end{lemma}
\begin{proof}
Consider an optimal schedule for any problem instance. For every job in the first rank that is performed on a machine for which the completion time after rank $k$ in the optimal schedule is less than the makespan of the optimal
schedule, increase the processing time so that the completion time after rank $k$ in the optimal schedule is equal to the makespan. Clearly, the increase in processing times does not lead to a violation of the rank restriction for
flowtime-optimal schedules. Also, the increase in processing times results in a rectangular optimal schedule. By Lemma \ref{NONDECREASING_PROFILE}, the increase in processing times does not affect the makespan of the optimal
schedule, and it may only increase the makespan of LD schedules. Therefore, it will either increase or leave unchanged the $t_{LD}/t^*$ ratio. \end{proof}

Note that, if the Coffman-Sethi conjecture is false, a problem instance with a worst-case $t_{{LD}}/t^*$ ratio would be a counterexample to the conjecture. This leads to the following corollaries.

\begin{corollary} \label{TYPE_R} If the Coffman-Sethi conjecture is false, then there exists a minimal counterexample to the conjecture of Type R.
\end{corollary}

\begin{corollary}
\label{cor:IR}
If the Coffman-Sethi conjecture is false, then there exists a minimal counterexample to the conjecture of Type IR.
\end{corollary}

\begin{proof} This follows from Corollary \ref{LD_TYPE_I} and the proof of Lemma \ref{lem:CS1}.\end{proof}

\begin{theorem} If the Coffman-Sethi conjecture is false, then every minimal counterexample to the conjecture of Type IR or I satisfies
\[
\frac{t_{LD}}{t^*} < \frac{k}{k-1}.
\] \end{theorem}
\begin{proof}
Proof for minimal counterexamples of Type IR:\\
Suppose that the Coffman-Sethi conjecture is false. Among all counterexamples of
Type IR, consider a minimal counterexample P1. Now we apply the following two-step process:\\
Step 1: Reduce each nonzero processing time by $1$ to construct a new problem instance P2. Clearly, the assignment of jobs in each rank to machines can be kept unchanged for the new LD schedule. The schedule that was originally an
optimal rectangular schedule for P1 will not be rectangular but will be optimal for P2 after the reduction in processing times. In the modified schedule, there are only $(k-1)$ jobs assigned to the one or more machines with a zero
processing time job in the last rank. This implies that the new optimal makespan is equal to $\left[t^* - (k-1)\right]$.\\
Step 2: For every job in rank $1$ of the optimal schedule for P2 that is processed on a machine with a completion time after rank $k$ that is less than the makespan, increase the processing time so that the completion time after
rank $k$ becomes equal to the makespan. This produces a problem instance P2R of Type IR. P2R will have an optimal makespan that is equal to the optimal makespan of P2. Utilizing Lemma \ref{NONDECREASING_PROFILE}, it follows that
the LD makespan of P2R cannot be less than the LD makespan of P2.\\
The optimal makespan of P2R is equal to $\left[t^{*} - (k-1)\right]$. For the LD schedule for P2R, there are two possibilities:
(i) $t_{LD}$ gets reduced to a value that is greater than or equal to $\left[t_{LD} - (k-1)\right]$.
This results in a new problem instance of Type IR, with a makespan ratio that is at least
\[
\frac{t_{LD} - (k-1)}{t^{*} - (k-1)}.
\]
This new ratio is larger than $t_{LD}/t^{*}$, thus contradicting the assumption that the original problem instance
was a minimal counterexample of Type IR.
(ii) $t_{LD}$ gets reduced to $\left(t_{LD} - k \right)$. This results in a new problem instance with a makespan ratio
at least
\[
\frac{t_{LD} - k}{t^{*} - (k-1)}.
\]
The original problem instance P1 was a minimal counterexample of Type IR. This implies that
$(t_{LD} - k)/(t^{*} - (k-1))$ is less than $t_{LD}/t^{*}$. Therefore, $t_{LD}/t^{*} < k/(k-1)$. This completes the proof for minimal counterexamples of Type IR.\\
Proof for minimal counterexamples of Type I:\\
Consider a minimal counterexample P1$'$ of Type I. Now reduce each nonzero processing time by $1$ to construct a new problem instance P2$'$ with integer processing times. Clearly, the assignment of jobs in each rank to machines can
be kept unchanged for the new LD schedule and the new optimal schedule. The new optimal makespan is less than or equal to $t^{*} - (k-1)$. For the new LD schedule, there are two possibilities:\\
(i) $t_{LD}$ gets reduced to a value that is equal to $\left[t_{LD} - (k-1)\right]$. This results in a new problem instance with a makespan ratio that is at least $\left[t_{LD} - (k-1)\right]/\left[t^{*} - (k-1)\right]$. This new
ratio is larger than $t_{LD}/t^{*}$, thus contradicting the assumption that the original problem instance was a minimal counterexample of Type I.\\
(ii) $t_{LD}$ gets reduced to $t_{LD} - k$. This results in a new problem instance with a makespan ratio that
is greater than or equal to $\left[t_{LD} - k\right]/\left[t^{*} - (k-1)\right]$.
The original problem instance P1$'$ was a minimal counterexample of Type I. This implies that
$(t_{LD} - k)/\left[t^{*} - (k-1)\right]$ is less than $t_{LD}/t^{*}$. Therefore, $t_{LD}/t^{*} < (k/(k-1))$. This completes the proof for minimal counterexamples of Type I.\end{proof}

The following result follows directly from the above theorem and from Propositions \ref{prop:2} and \ref{lem:4.3}.

\begin{corollary}
\label{cor:CS-restricted-rank}
The Coffman-Sethi conjecture is true, if it is true for $k \in \{3,4,5,6\}$.  In particular,
\begin{itemize}
\item
for $m \geq 4$, settling the cases $k \in \{3,4,5\}$ suffices;
\item
for $m=3$, settling the cases $k \in \{4,5,6\}$ suffices.
\end{itemize}
\end{corollary}

We define a minimal counterexample of Type IR1 as follows. If the Coffman-Sethi conjecture is false, a minimal counterexample of Type IR1 is a minimal counterexample of Type IR that has an LD schedule with the following property:
Every machine with a completion time after rank $k$ equal to the makespan has a job with processing time equal to $\lambda_{k}$ in rank $k$, where $k$ denotes the number of ranks.

\begin{lemma} \label{IR1_EXISTS} If the Coffman-Sethi conjecture is false, then there exists a minimal counterexample to the conjecture of Type IR1, and every minimal counterexample of Type IR is a minimal counterexample of Type
IR1.\end{lemma}

\begin{proof} Suppose the Coffman-Sethi conjecture is false.
Then, by Corollary \ref{cor:IR}, a counterexample of Type IR
exists. Suppose, for a contradiction, that there exists a minimal
counterexample P1 of Type IR that is not Type IR1.
Now, apply the following two-step process.

\noindent
Step 1: Construct a new problem instance P2 as follows. Subtract $1$ time unit from the processing time of every job in rank $(k-1)$. Also subtract $1$ time unit from the processing time of every job in rank $k$ that has a processing
time of $\lambda_{k}$. Note that, for a minimal counterexample, all processing times in rank $k$ are not equal, therefore there exist processing times in rank $k$ that are less than $\lambda_{k}$. Leave these processing times
unchanged. Note that the assignment of jobs in each rank to machines in an LD schedule remains unchanged. Thus, P2 has an LD
schedule with objective value $(t_{LD}-1)$, where $t_{LD}$ is the objective value of the LD schedule for P1.
Let $t^{*}$ denote the optimal makespan of problem instance P1.  Then, the problem P2 has an optimal makespan that is equal to $(t^{*} - 1)$.\\
Step 2: For every job in rank $1$ of the optimal schedule for P2 that is processed on a machine with a completion time after rank $k$ that is less than the makespan, increase the processing time so that the completion time after
rank $k$ becomes equal to the makespan.\\

This produces a problem instance P2R of Type IR. By Lemma \ref{NONDECREASING_PROFILE}, the LD makespan of P2R cannot be less than the LD makespan of P2.  So, it is
at least $(t_{LD}-1)$.  The optimal makespan for P2R is $(t^*-1)$ by construction.  Thus, P2R gives a strictly worse approximation ratio than P1, a contradiction to the minimality of P1.
Therefore, P1 (as well as any other minimal counterexample of type IR) is a minimal counterexample of Type IR1.
\end{proof}

We define a problem instance, a counterexample and a minimal counterexample of \emph{Type I1} as follows. A problem instance of \emph{Type I1} is a problem instance of Type I that has an LD schedule with the following properties:
\begin{enumerate}[(i)]
\item
It has only one machine $i'$ with a completion time after rank $k$ equal to the makespan.\\
\item
Machine $i'$ has a processing time equal to $\lambda_{k-1}$ in rank $(k-1)$ and $\lambda_k$ in rank $k$.
\end{enumerate}

If the Coffman-Sethi conjecture is false, a counterexample to the conjecture of \emph{Type I1} is a counterexample of Type I that has an LD schedule with the
above-mentioned properties.

\begin{lemma}
\label{lem:I1}
If the Coffman-Sethi conjecture is false, then there exists a minimal counterexample to the conjecture of Type I1. \end{lemma}

\begin{proof} Suppose the Coffman-Sethi conjecture is false.
Then, by Lemma \ref{IR1_EXISTS}, there exists a minimal counterexample of Type IR1.  Call this instance P1.
Now, construct a new problem instance P2 as follows. Subtract $1$ time unit from the processing time of every job in rank $(k-2)$. Also subtract $1$ time unit from the processing time of every job in rank $(k-1)$ that has a
processing time of $\lambda_{k-1}$. Note that, for a minimal counterexample, all processing times in rank $(k-1)$ are not equal, therefore there exist processing times in rank $(k-1)$ that are less than $\lambda_{k-1}$. Leave these
processing times unchanged. Note that the assignment of jobs in each rank to machines in the LD schedule remains unchanged, except possibly in rank $k$. Problem instance P1 had an optimal rectangular schedule. Reducing the
processing time of every job in rank $(k-2)$ by $1$ leaves the rectangular property unchanged and results in a reduction of $1$ in the optimal makespan. A further reduction of $1$ in one or more, but not all, jobs in rank $(k-1)$
results in no further reduction in the optimal makespan. Thus, problem P2 has an optimal makespan equal to $(t^{*} - 1)$. The new LD objective value is either $(t_{LD}-1)$ (if  there is at least one machine whose completion time
originally equaled $t_{LD}$ and after the modification of the processing times, now equals $(t_{LD}-1)$ since in rank $(k-1)$ it had a job with processing time less than $\lambda_{k-1}$)
or $(t_{LD}-2)$ (if every machine, with completion time equal to $t_{LD}$ for P1, had a job in rank $(k-1)$ with processing time $\lambda_{k-1}$).  We will show below that the former case cannot happen.

Now, construct a problem instance P2R of Type IR by adding $1$ unit to the processing time of every job in rank $1$ that is performed on a machine with completion time after rank $k$ in the optimal schedule that is less than the
makespan. The optimal makespan of P2R is equal to the optimal makespan of P2, $(t^*-1)$.  Since P1 is a minimal counterexample of Type IR, using Lemma \ref{NONDECREASING_PROFILE} and the above argument,
we deduce that the makespan of the LD schedule for problem instance P2R is $(t_{LD}-2)$.
Thus, in the LD schedule for P1, every machine with completion time equal to $t_{LD}$ had a job in rank $(k-1)$ with processing time $\lambda_{k-1}$.  By Lemma \ref{IR1_EXISTS},
every machine with completion time equal to $t_{LD}$ has jobs with processing times $\lambda_{k-1}$ and $\lambda_k$ in ranks $(k-1)$ and $k$ respectively.  Pick one these machines, call it $i'$.

Construct a new counterexample P3 as follows. In the LD schedule for P1, for every machine $i \neq i'$ with a completion time after rank $k$ equal to the makespan, delete the job in rank $k$. The makespan and the set of jobs
assigned to $i'$ in the LD schedule for P3 are the same as those in the LD schedule for P1. The makespan of the optimal schedule for P3 is less than or equal to the makespan of the optimal schedule for P1. However, the optimal
schedule for P3 may not be rectangular. P3 is clearly a counterexample to the conjecture of Type I1. It follows that there exists a minimal counterexample to the conjecture of Type I1. \end{proof}

We define a problem instance, a counterexample and a minimal counterexample of \emph{Type I2} as follows. A problem instance of \emph{Type I2} is a problem instance of Type I1 with the following property
in an LD schedule: The machine $i'$ with a completion time equal to the makespan has a processing time equal to $\lambda_r$ in rank $r$ for every $r \in \{2,3, \ldots, k\}$.

If the Coffman-Sethi conjecture is false, a counterexample to the conjecture of \emph{Type I2} is a counterexample of Type I1 with the above-mentioned property.

\begin{lemma} \label{I2_EXISTS} If the Coffman-Sethi conjecture is false, then there exists a minimal counterexample to the conjecture of Type I2. \end{lemma}

\begin{proof} Suppose that the Coffman-Sethi conjecture is false.
We proceed by induction on the number of ranks for which the claim holds.  The base case is given
by Lemma \ref{lem:I1}. The induction hypothesis is that there exists a minimal counterexample P1 of Type I1 that has an LD schedule with a machine $i'$ which has a processing time equal to $\lambda_r$ in rank
$r$ for every $r \in \{h, h+1, \ldots, k\}$, where $h \geq 3$. If machine $i'$ has a job in rank $(h-1)$ with processing time $\lambda_{h-1}$ then we are done;
otherwise, construct P2 from P1 by subtracting $1$ time unit from the processing time of every job in rank $(h-2)$, and subtracting $1$
time unit from the processing time of every job in rank $(h-1)$ that has a processing time of $\lambda_{h-1}$. Leave the remaining processing times unchanged.
Note that either a new LD schedule assigns the same jobs to machine $i'$, or another machine $i''$ which had the same completion time as machine $i'$ after rank $(h-1)$ and a job with processing time
$\lambda_{h-1}$ now has a completion time one less and has the jobs originally scheduled on machine $i'$ for ranks $r \in \{h, h+1, \ldots, k\}$. In the former case, P2 has a strictly
worse approximation ratio than P1, a contradiction.  Therefore, we must be in the latter case.  In the latter case, go back to instance P1.  After rank $(h-1)$, both machines $i'$ and $i''$
have the same completion time.  For the ranks $h, h+1, \ldots, k$, swap all jobs on the machines $i'$ and $i''$.  We still have an
LD schedule for P1 with makespan $t_{LD}$; moreover, machine $i''$ is the unique machine with completion time equal to the makespan.
Finally, the processing times on machine $i''$ are $\lambda_{h-1}, \lambda_h, \lambda_{h+1}, \ldots, \lambda_k$ as desired.
It follows that there exists a minimal counterexample to the conjecture of Type I2. \end{proof}

\begin{lemma} \label{MU_1_IN_RANK_1} If the Coffman-Sethi conjecture is false, in a minimal counterexample of Type I2, the sole machine $i'$ with a completion time after rank $k$ equal to the makespan
in the LD schedule has a processing time equal to $\mu_1$ in rank $1$. \end{lemma}

\begin{proof} Suppose the Coffman-Sethi conjecture is false.  Then, by the previous lemma, there exist one or more minimal counterexamples to the conjecture of Type I2. Let P1 denote one of these minimal counterexamples. If, in the
LD schedule for P1, there exists only one machine with a processing time equal to $\lambda_{r}$ in rank $r$ for $r \in \{2,3, \ldots, k\}$, the lemma clearly holds.
Assume that there exists a set of two or more machines with a processing time equal to $\lambda_{r}$ in rank $r \in \{2,3, \ldots, k\}$. We may assume that machine $i'$
has a job with processing time strictly greater than $\mu_1$ in rank 1.  Then, there exists machine $i''$ which has processing time
$\mu_1$ in rank one, then by the definition of the LD algorithm, machine $i''$ must have the processing times:
$\mu_1, \lambda_2, \lambda_3, \ldots, \lambda_k$ respectively.  We delete all the jobs on machine $i''$ and delete the machine $i''$
to generate a new instance P2.

Note that the LD schedule for P2 is unchanged on machine $i'$ and all the other machines except $i''$.  Therefore,
$t_{LD}$ is unchanged.  Next, we prove that the optimal makespan for P2 is no larger than that for P1.  Consider an
optimal schedule $\S^*$ for P1.  A machine has the jobs with processing times
$\mu_1, p_2, p_3, \ldots, p_k$ respectively. Clearly, $p_r \leq \lambda_r$, for every $r \in \{2,3, \ldots, k\}$.
If the equality holds throughout, then we are done.  Otherwise, for each $r$ that the inequality is strict,
we find the job with processing time $\lambda_r$ in rank $r$ and swap it with $p_r$.  These operations may increase
the completion time on machine $i_1$; however, they will not increase it on any other machine.  At the end, we delete
the machine (with the processing times $\mu_1, \lambda_2, \lambda_3, \ldots, \lambda_k$).
What remains is
a feasible schedule for instance P2 whose makespan is at most the optimal makespan for P1.

We repeat the above procedure, until there exists only one machine in the LD schedule with a processing time equal to $\lambda_{r}$ in rank $r \in \{2,3, \ldots, k\}$.
From the mechanics of the LD algorithm, it follows that no other
machine has a smaller processing time in rank $1$. This completes the proof.

\end{proof}

\begin{lemma} \label{SMALLEST_COMPLETION_TIME} If the Coffman-Sethi conjecture is false, then in the LD schedule for a minimal counterexample of Type I2, the smallest completion time after rank $k$ on any machine is
at least
\[t_{LD} - \displaystyle{\max_{r \in \{2, 3, \ldots ,k\}} \left\{ \lambda_r - \mu_r \right\}}.
\]
\end{lemma}

\begin{proof} Suppose the Coffman-Sethi conjecture is false.
Then, by Lemma \ref{I2_EXISTS}, there exists a minimal counterexample of Type I2. By Lemma \ref{MU_1_IN_RANK_1}, in a minimal counterexample of Type I2,
the sole machine $i'$ with a completion time after rank $k$ equal to the makespan has a processing time equal to $\mu_1$ in rank $1$.
For $i \in\{1,2, \ldots, m \}$, $i \neq i'$, let $r_i$ denote the smallest value of $r \in \{1,2, \ldots, k\}$, for which the completion time
after rank $r_i$ on machine $i$ is less than the completion time after rank $r_i$ on machine $i'$. By Lemma \ref{I2_EXISTS}, it follows that there exists a value $r_i \leq k$ for all $i \neq i'$.
Lemma \ref{MU_1_IN_RANK_1} implies that $r_i \geq 2$. There are two possible cases:
\begin{enumerate}[(a)]
\item
$r_i = k$. In this case, the completion time after rank $k$ on machine $i$ is greater than or equal to $t_{LD} - \lambda_k$.
\item
$r_i < k$. In this case, from the mechanics of the LD algorithm, it is evident that the processing time of the job on machine $i$ in rank $(r_i+1)$
and all of the following ranks must be equal to $\lambda_r$.
Therefore, the completion time after rank $k$ on machine $i$ is greater than or equal to $t_{LD} - \left( \lambda_{r_i} - \mu_{r_i} \right)$.
\end{enumerate}
This completes the proof of the lemma. \end{proof}

\begin{lemma} \label{LEADING_AFTER_RANK_k-1} If the Coffman-Sethi conjecture is false, then in the LD schedule for any minimal counterexample of Type I2, there exists at least one value of $i''$ that satisfies $i'' \neq i'$, for which
the completion time after rank $(k - 1)$ on machine $i''$ is greater than or equal to the completion time after rank $(k - 1)$ on machine $i'$, where $i'$ denotes the sole machine with a completion time after rank $k$ equal to the
makespan.\end{lemma}

\begin{proof} Assume that the claim of the lemma does not hold (we are seeking a contradiction). Therefore, there exists a minimal counterexample of Type I2 for which every machine $i \neq i'$ has a completion time after rank
$(k-1)$ that is less than the completion time after rank $(k-1)$ on machine $i'$. From the mechanics of the LD algorithm, it is evident that every machine $i \neq i'$ has a processing time in rank $k$ that is greater than or equal
to the processing time of the job assigned to machine $i'$ in rank $k$. Hence every machine has a job with processing time $\lambda_k$ assigned to it in rank $k$. Therefore all jobs in rank $k$ can be removed from the
counterexample to obtain a smaller counterexample of Type I2 with a larger $t_{LD}/t^*$ ratio. This contradicts the assumption that the original counterexample was a minimal counterexample of Type I2.\end{proof}

\begin{lemma} \label{LEADING_AFTER_RANK_k} If the Coffman-Sethi conjecture is false, then in the LD schedule for any minimal counterexample of Type I2, there exists at least one value of $i''$ that satisfies $i'' \neq i'$, for which the
completion time after rank $k$ on machine $i''$ is less than the completion time after rank $(k - 1)$ on machine $i'$, where $i'$ denotes the sole machine with a completion time after rank $k$ equal to the makespan.\end{lemma}

\begin{proof} Assume that the claim of the lemma does not hold (we are seeking a contradiction) and there exists a minimal counterexample of Type I2 for which every machine $i \neq i'$ has a completion time after rank $k$ that is
greater than or equal to the completion time after rank $(k-1)$ on machine $i'$. A minimal counterexample must have at least $3$ ranks. We first prove the lemma for the $k \geq 4 $ case and then prove the lemma for the $k = 3$
case.  We have $t^* \geq t_{LD} - \lambda_k + \frac{\lambda_k}{m}$.   Supposing that we have a counterexample, we deduce

\bea
  t^* & > & \left(\frac{5m-2}{4m-1}\right)t^* - \left(\frac{m-1}{m}\right)\lambda_k \nonumber\\
  \label{A1} t^* & < & \left(4 - \frac{1}{m}\right)\lambda_k.
\eea
For $k \geq 4$, $t^* > 4 \lambda_k$, we reached a contradiction.  This completes the proof of the lemma for $k  \geq 4$.

For $k = 3$, there exists at least one machine in the optimal schedule with a job with processing time of $\lambda_3$ in the third rank.
Therefore,
\[
t^* \geq \mu_1 + \mu_2 + \lambda_3.
\]
From the preceding lemmas, it follows that $t_{LD} = 2\lambda_2 + \lambda_3.$
For a counterexample, we must have
\[ \frac{t_{LD}}{t^*} > \frac{5m-2}{4m-1}.
\]
Therefore,
\[
\frac{2\lambda_2 + \lambda_3}{\lambda_2 + 2\lambda_3} > \frac{5m-2}{4m-1}.
\]
It follows that
\bea
  \label{B1} \lambda_2 & > & \left(2 - \frac{1}{m}\right) \lambda_3.
\eea
Utilizing the inequality $t^* \geq t_{LD} - \lambda_k + \frac{\lambda_k}{m},$
and the fact $t_{LD}=2 \lambda_2+\lambda_3$, we deduce
\bea
  \label{C1} t^* & \geq & (2\lambda_2 + \lambda_3) - \lambda_3 + \frac{\lambda_3}{m}.
\eea
From inequalities \eqref{B1} and \eqref{C1}, it follows that
  \bea \label{D1} t^* & \geq & \left(4 - \frac{1}{m}\right)\lambda_3.
\eea
From inequalities \eqref{A1} and \eqref{D1}, we have a contradiction for $k = 3$.
This completes the proof of the lemma. \end{proof}

\begin{theorem} \label{PROOF_FOR_THREE_RANKS} The Coffman-Sethi conjecture holds for $k$ equal to $3$.\end{theorem}

\begin{proof}
Suppose the statement of the theorem is false.
Consider a minimal counterexample of Type I2 with $k=3$.
Consider the LD schedule for this minimal counterexample.
Let $i'$ denote the sole machine with a completion time after rank $k$ equal to the makespan.
In the LD schedule, let $M_1$ denote the set of $m_1$ machines with a completion time after rank $(k-1)$ that is greater than or equal to the completion time after rank $(k-1)$ on machine $i'$. Note that set $M_1$ includes machine $i'$
and that machine $i'$ is the only machine in $M_1$ with a processing time of $\lambda_3$ in the third rank. The remaining $(m_1 - 1)$ machines have no job assigned to them in the third rank.

From Lemma \ref{LEADING_AFTER_RANK_k-1}, it is evident that $m_1 \geq 2$.
Let $M_2$ denote the set of $m_2 := m - m_1$ machines with a completion time after rank $(k-1)$ that is less than the completion time after rank $(k-1)$ on machine $i'$. Every machine in $M_2$ has a processing time of $\lambda_3$ in the
third rank. The total number of machines with $\lambda_3$ in the third rank is $(m -  m_1 + 1)$.

From Lemma \ref{LEADING_AFTER_RANK_k}, it is evident that $m_2 \geq 1$.
Select a machine $i_2$ in the set $M_2$ with a completion time after rank $k$ that is less than the completion time after rank $(k-1)$ on machine $i'$.
Let $\alpha_1$ denote the processing time in rank $1$ of machine $i_2$ and let $\alpha_2$ denote the processing time in rank $2$ of machine $i_2$.

For any machine $i_1$ in the set $M_1$:
Let $\beta_1$ denote the processing time in rank $1$ of machine $i_1$ and let $\beta_2$ denote the processing time in rank $2$ of machine $i_1$.
Note that $\beta_1 + \beta_2 \geq 2 \lambda_2 > \alpha_1 + \alpha_2 + \lambda_3$.

Two possible cases need to be considered.

\underline{Case 1:} $\beta_1 \leq \alpha_1$

In this case, $\beta_2 \geq 2\lambda_2 - \beta_1 > \alpha_1 + \alpha_2 + \lambda_3 - \beta_1$.
Hence,\\
$\beta_2 > \alpha_2 + \lambda_3 \geq \mu_2 + \lambda_3$.
Therefore, $\beta_2 > 2\lambda_3$.

\underline{Case 2:} $\beta_1 > \alpha_1$

From the dynamics of the LD algorithm, $\beta_2 \leq \alpha_2$.
In this case,\\
$\beta_1 \geq 2\lambda_2 - \beta_2 > \alpha_1 + \alpha_2 + \lambda_3 - \beta_2$,
whence $\beta_1 > \alpha_1 + \lambda_3 \geq \mu_1 + \lambda_3$.
Therefore, $\beta_1 > \lambda_2 + \lambda_3$.

It is evident that, in any schedule (including the optimal schedule),
one of the following must be true:\\
Case A: There exists a machine with a processing time of $\lambda_3$ in rank $3$ and a processing time greater than $2\lambda_3$ in rank $2$.\\
Case B: There exists a machine with a processing time of $\lambda_3$ in rank $3$ and a processing time greater than $(\lambda_2 + \lambda_3)$ in rank $1$.\\
Case C: There exists a machine with a processing time greater than $2\lambda_3$ in rank $2$ and a processing time greater than $(\lambda_2 + \lambda_3)$ in rank $1$.\\
In Case A, the makespan must be greater than $\lambda_3 + 2\lambda_3 + \mu_1$.\\
In Case B, the makespan must be greater than $\left[\lambda_3 + \mu_2 + (\lambda_2 + \lambda_3)\right]$.\\
In Case C, the makespan must be greater than $\left[2\lambda_3 + (\lambda_2 + \lambda_3)\right]$.

Thus, in all three cases, the makespan must be greater than $(\lambda_2 + 3\lambda_3)$.
Therefore, $t^* > \lambda_2 +3 \lambda_3.$  By Lemma~\ref{MU_1_IN_RANK_1}, $t_{LD} = 2\lambda_2+\lambda_3$.
Thus, we have
\[
\frac{5m-2}{4m-1} < \frac{t_{LD}}{t^*} < \frac{2\lambda_2 + \lambda_3}{\lambda_2 + 2\lambda_3}.
\]

In a minimal counterexample of Type I2, $m \geq 3$.  We have
\[
\frac{13}{11} < \frac{2\lambda_2 + \lambda_3}{\lambda_2 + 2\lambda_3}.
\]
Therefore,
\bea
\label{A2} \frac{\lambda_2}{\lambda_3} & > & \frac{28}{9}.
\eea
Also, note that, since $t^* > \mu_1+\lambda_2 + \mu_3$, we have
\[
\frac{13}{11} < \frac{t_{LD}}{t^*} < \frac{2\lambda_2 + \lambda_3}{\mu_1 + \lambda_2 + \mu_3}.
\]
Therefore,
\bea
\label{B2} \frac{\lambda_2}{\lambda_3} & < & \frac{11}{4}.
\eea
From inequalities \eqref{A2} and \eqref{B2}, we obtain a contradiction. This completes the proof of the theorem.
\end{proof}

\begin{theorem} \label{PROOF_FOR_THREE_MACHINES} The Coffman-Sethi conjecture holds for $m=3$.\end{theorem}
\begin{proof}
Consider a minimal counterexample of Type I2 with $m$ equal to $3$.
From Lemmas \ref{LEADING_AFTER_RANK_k-1} and \ref{LEADING_AFTER_RANK_k}, it follows that, for this counterexample, the following statements hold
for an LD schedule.
(i) There is one machine $i'$ with a completion time after rank $k$ that is equal the makespan.
(ii)There is one machine $i''$ with a completion time after rank $k$ that is less than the makespan and a completion time after rank $(k-1)$ that is greater than or equal to the completion time after rank $(k-1)$ on machine $i'$.
(iii) There is one machine $i'''$ with a completion time after rank $k$ that is less than the completion time after rank $(k-1)$ on machine $i'$.

From the above given statements and from Lemma \ref{SMALLEST_COMPLETION_TIME}, it follows that
\bea
\label{A3}
3t^* \geq t_{LD} + (t_{LD} - \lambda_k) +
\left(t_{LD} - \max_{r \in \{2, 3, \ldots ,k\}} \left\{ \lambda_r - \mu_r \right\}\right).
\eea
Considering machine  $i'''$ , and Lemma \ref{SMALLEST_COMPLETION_TIME}, it is clear that
\bea
\max_{r \in \{2, 3, \ldots ,k\}} \left\{ \lambda_r - \mu_r \right\} > \lambda_k. \nonumber
\eea
Therefore,
\bea
\max_{r \in \{2, 3, \ldots ,k\}} \left\{ \lambda_r - \mu_r \right\} = \max_{r \in \{2, 3, \ldots ,k-1\}} \left\{ \lambda_r - \mu_r \right\}.\nonumber
\eea
It follows that
\bea
\label{B3}
\max_{r \in \{2, 3, \ldots ,k-1\}} \left\{ \lambda_r - \mu_r \right\} \leq \lambda_2 - \mu_{k-1} = \lambda_2-\lambda_k.
\eea
From \eqref{A3} and \eqref{B3}, it follows that
\bea
3t^* \geq t_{LD} + (t_{LD} - \lambda_k) + \left[t_{LD} - (\lambda_2 - \lambda_k)\right].
\eea
Therefore,
\bea
\label{C3}
t^* \geq t_{LD} - \frac{\lambda_2}{3}.
\eea
For a counterexample to the Coffman-Sethi conjecture with $m$ equal to $3$, we must have
\bea
\label{D3}
\frac{t_{LD}}{t^*} > \frac{13}{11}.
\eea
From inequalities \eqref{C3} and \eqref{D3}, it follows that
\bea
\label{E3}
t^* < \frac{11}{6}\lambda_2.
\eea

For two or more ranks, $t^*$ must be at least equal to $\mu_1 + \lambda_2$ or $2\lambda_2$. Therefore, inequality \eqref{E3} cannot hold and we have a contradiction. This completes the proof of the lemma.
\end{proof}

Theorems \ref{PROOF_FOR_THREE_RANKS} and \ref{PROOF_FOR_THREE_MACHINES} allow us to strengthen the previous corollary.

\begin{corollary}
\label{cor:CS-restricted-rank2}
The Coffman-Sethi conjecture is true, if it is true for $m \geq 4$ and for $k \in \{4,5\}$.
\end{corollary}

\section{Conclusion and Future Research}

We studied the structure of potential minimal counterexamples to Coffman-Sethi (1976) conjecture
and proved that the conjecture holds for the cases (i) $m=2$, (ii) $m=3$, and (iii) $k=3$. We also proved
that to establish the correctness of the conjecture,
it suffices to prove the conjecture for $m \geq 4$ and $k \in \{4,5\}$.  Moreover, for each $m \geq 4$, these
remaining cases $k \in \{4,5\}$ can be verified by employing our approach from Section 4 and the Appendix,
via solving a finite number of LP problems.

As a by-product of our approach, we introduced various techniques to analyze worst case
performance ratios.  These techniques may be useful in analyzing the performance
of approximation algorithms for other scheduling problems, or in general, other combinatorial
optimization problems of similar nature.

{\bf Acknowledgment:}
The work was supported in part by Discovery Grants from NSERC (Natural Sciences and Engineering Research Council of Canada) as well as
an Undergraduate Research Assistantship from NSERC.

\section{Appendix}

\subsection{Three-Machine Case}

For $r \in \{1,2, 3\}$, let $\lambda_r$, $\alpha_r$ and $\mu_r$ denote the processing times in rank $r$, where $\lambda_r \geq \alpha_r \geq \mu_r$.
Clearly, $\mu_1 = \lambda_2$ and $\mu_2 = \lambda_3$ and $\mu_3 = 0$ and either $\alpha_3 = \lambda_3$ or $\alpha_3 = \mu_3=0$.

The first two ranks of the LD schedule will look like
\[
\left(%
\begin{array}{cc}
  \lambda_1 & \lambda_3 \\
  \alpha_1 & \alpha_2 \\
  \lambda_2 & \lambda_2 \\
\end{array}%
\right).
\]
The third rank will fit depending on the length of processing times on the
first two ranks.  We will use case analysis to cover all possible LD schedules.

For all cases we will have $\lambda_1 \ge \alpha_1 \ge \lambda_2 \ge
\alpha_2 \ge \lambda_3 > 0$.  This results in the following 5 constraints.
\[
\begin {array}{cccccccr}
    -\lambda_1 & & & + \alpha_1 & & \le & 0 & (1)\\
    & \lambda_2 & & - \alpha_1 & & \le & 0  & (2)\\
    & -\lambda_2 & & & + \alpha_2 & \le & 0  & (3)\\
    & & \lambda_3 & & - \alpha_2 & \le & 0  & (4)\\
    & & -\lambda_3 & & & \le & 0 & (5)\\
\end {array}
\]

We begin by looking at the cases when $\alpha_3 = \lambda_3$.

Consider the case $\lambda_1 + \lambda_3 \ge \alpha_1 + \alpha_2 \ge 2\lambda_2$,
then we have the constraints
\[
\begin {array}{cccccccr}
    -\lambda_1 & & - \lambda_3 & + \alpha_1 & \alpha_2 & \le & 0 & (6)\\
    & 2\lambda_2 & & - \alpha_1 & - \alpha_2 & \le & 0 & (7)\\
\end {array}
\]

Further, the LD schedule will look like
\[
\left(%
\begin{array}{ccc}
  \lambda_1 & \lambda_3 & 0\\
  \alpha_1 & \alpha_2 & \lambda_3\\
  \lambda_2 & \lambda_2 & \lambda_3\\
\end{array}%
\right)
\]

$t_{LD} = \max (\lambda_1+\lambda_3 , \alpha_1 + \alpha_2 + \lambda_3)$.

If $t_{S^*} = \lambda_1 + \lambda_3$, then $t_{S^*}$ is equal to a lower bound and $\frac {t_{LD}}{t_{S^*}} = 1$.  So, we may assume $t_{S^*} = \alpha_1 + \alpha_2 + \lambda_3$.

Consider an optimal configuration for this problem. For the third job, the machine with $\lambda_1$ in the first rank must have a job in the third rank with processing time $0$, otherwise that machine will have processing time
at least
$\lambda_1 + 2\lambda_3 \geq \alpha_1 + \alpha_2 + \lambda_3 = t_{LD}$.  Since the machine with $\lambda_1$ in the frst rank has a job with processing time $0$ in the last rank, the machine with $\alpha_1$ in the first
rank has a job with processing time $\lambda_3$ in the third rank. In order for $t_{S^*}$ to be less than $t_{LD}$, the machine with $\alpha_1$ in the first rank must have a job with processing time $\lambda_3$ in the second
rank as well.  Thus, the
optimal configuration must be either
\[
\left(%
\begin{array}{ccc}
  \lambda_1 & \alpha_2 & 0\\
  \alpha_1 & \lambda_3 & \lambda_3\\
  \lambda_2 & \lambda_2 & \lambda_3\\
\end{array}%
\right) \mbox { or }
\left(%
\begin{array}{ccc}
  \lambda_1 & \lambda_2 & 0\\
  \alpha_1 & \lambda_3 & \lambda_3\\
  \lambda_2 & \alpha_2 & \lambda_3\\
\end{array}%
\right).
\]

We consider the case where the former is the optimal configuration.
Suppose $t_{S^*} = \lambda_1 + \alpha_2$.  Then, we have $\lambda_1 + \alpha_2 \geq \alpha_1 + 2\lambda_3$ and $\lambda_1 + \alpha_2 \geq 2\lambda_2 + \lambda_3$.

This leads to two more constraints
\[
\begin {array}{cccccccr}
    -\lambda_1 & & + 2\lambda_3 & + \alpha_1 & - \alpha_2 & \le & 0, & (8)\\
    -\lambda_1 & + 2\lambda_2 & + \lambda_3  & & - \alpha_2 & \le & 0. & (9)\\
\end {array}
\]
Since the processing times of the jobs can be scaled by any positive number
(here, we are no longer requiring integer or even rational processing times
in the input data), we let
$t_{S^*} = 1$.  We add the constraint
\[
\begin {array}{cccccr}
    \lambda_1 & + & \alpha_2 & \le & 1. & (10)\\
\end {array}
\]
We now consider the linear program
\[
\begin {array} {ll}
    \max & t_{LD} = \alpha_1 + \alpha_2 + \lambda_3 \\
    \textup{subject to:} & (1),\; (2),\; \cdots, (10).
\end {array}
\]
The solution will give the worst case ratio of $\frac {t_{LD}}{t_{S^*}}$ for this particular case.

Solving the LP yields objective value $\frac {9}{8}$ when
[$\lambda_1, \lambda_2, \lambda_3, \alpha_1, \alpha_2$] = $\frac {1}{8}$
[5, 3, 2, 4, 3].

We now suppose $t_{S^*} = \alpha_1 + 2\lambda_3$.  Constraints (8)--(10) are
now replaced with
\[
\begin {array}{cccccccr}
    \lambda_1 & & - 2\lambda_3 & - \alpha_1 & + \alpha_2 & \le & 0 & (8)\\
    & 2\lambda_2 & - \lambda_3 & - \alpha_1 & & \le & 0 & (9)\\
    & & 2\lambda_3 & + \alpha_1 & & \le & 1. & (10)\\
\end {array}
\]
This new LP also yields objective value $\frac {9}{8}$ when
[$\lambda_1, \lambda_2, \lambda_3, \alpha_1, \alpha_2$] = $\frac {1}{8}$
[5, 3, 2, 4, 3].

We now suppose $t_{S^*} = 2\lambda_2 + \lambda_3$.  Constraints (8)--(10) are
replaced with
\[
\begin {array}{cccccccr}
    \lambda_1 & - 2\lambda_2 & - \lambda_3 & & + \alpha_2 & \le & 0 & (8)\\
    & -2\lambda_2 & + \lambda_3 & + \alpha_1 & & \le & 0 & (9)\\
    & 2\lambda_2 & + \lambda_3 & & & \le & 1. & (10)\\
\end {array}
\]
Again the new LP yields objective value $\frac {9}{8}$ when
[$\lambda_1, \lambda_2, \lambda_3, \alpha_1, \alpha_2$] = $\frac {1}{8}$
[5, 3, 2, 4, 3].

We now consider the other possible configuration as the optimal configuration.
We consider each possibility for $t_{S^*}$ by modifying the
constraints (8)--(10) to accommodate each $t_{S^*}$.  For example, when
$t_{S^*} = \lambda_1 + \lambda_2$,  constraints (8) - (10) are
\[
\begin {array}{cccccccr}
    -\lambda_1 & - \lambda_2 & + 2\lambda_3 & + \alpha_1 & & \le & 0 & (8)\\
    -\lambda_1 & & + \lambda_3 & & + \alpha_2 & \le & 0 & (9)\\
    \lambda_1 & + \lambda_2 & & & & \le & 1. & (10)\\
\end {array}
\]

For all three LPs for this configuration, the optimal objective value is
$\frac{9}{8}$ which is attained when
[$\lambda_1, \lambda_2, \lambda_3, \alpha_1, \alpha_2$] = $\frac {1}{8}$
[5, 3, 2, 4, 3].

We have exhausted this case and now move to the case where
$\lambda_1+\lambda_3 \ge 2\lambda_2 \ge \alpha_1+\alpha_2$.
We modify the constraints (6) and (7) to reflect the inequalities
related to this case.

We continue in this manner and set up and solve LP problems for
every possible case for the $3$-machine, $3$-rank problem. In every case, the makespan ratio is less than or equal to $13/11$. This shows that the Coffman-Sethi conjecture holds for
the $3$-machine, $3$-rank problem. Further, an LD makespan of $13/11$ is actually attained in some cases, thus showing, one more time, that the Coffman-Sethi conjecture provides a tight bound for the $3$-machine, $3$-rank problem.
If the size of a hypothesised minimal counterexample to the Coffman-Sethi conjecture can be shown the satisfy $m \leq 3$ and $k \leq 3$, then the above analysis would imply that the Coffman-Sethi conjecture holds for all problem
instances.


\begin{thebibliography}{99}
\bibitem{BCS1974}
Bruno, J., E. G. Coffman and R. Sethi, Algorithms for minimizing mean flow time, \emph {INFORMATION PROCESSING '74: Processings of the IFIPS Conference}, 504--510, North-Holland Publishing Company, 1974.
\bibitem{Cao1995}
Cao, F., Determining the performance ratio of algorithm MULTIFIT for scheduling, in \emph{Nonconvex Optimization Applications, Volume 4: Minimax and applications}, Kluwer Academic Publishers, Dordrecht, 79--96, 1995.
\bibitem{ChangHwang1999}
Chang, S. Y. and H. C. Hwang, The worst-case analysis of the MULTIFIT algorithm for scheduling nonsimultaneous parallel machines, \emph{Discrete Applied Mathematics}, 92(2--3), 135--147, 1999.
\bibitem{CoffmanGareyJohnson1978}
Coffman, E. G. M. R. Garey and D. S. Johnson, An application of bin-packing to multiprocessor scheduling. \emph {SIAM Journal of Computing} 1, 1--17, 1978.
\bibitem{CoffmanGareyJohnson1983}
Coffman, E. G., M. R. Garey and D. S. Johnson, Dynamic bin packing, \emph {SIAM Journal of Computing}, 12, 227--258, 1983.
\bibitem{CoffmanSethi1976a}
Coffman, E. G. and R. Sethi, Algorithms minimizing mean flowtime: schedule-length properties, \emph{Acta Informatica} 6, 1--14, 1976.
\bibitem{CoffmanSethi1976b}
Coffman, E. G. and R. Sethi, A generalized bound on LPT sequencing, \emph {Proceedings of the 1976 ACM SIGMETRICS Conference on Computer Performance Modeling Measurement and Evaluation}, 306--310, 1976.
\bibitem{CoffmanYannakakis1980}
Coffman, E. G. and M. Yannakakis, Permuting elements within columns of a matrix in order to minimize maximum row sum, \emph {Mathematics of Operations Research}, 9(3), 384--390, 1984.
\bibitem{ConwayMaxwellMiller1967}
Conway, R. W., W. L. Maxwell and L. W. Miller, \emph{Theory of Scheduling}, Addison Wesley, MA, USA, 1967.
\bibitem{Dosa2000}
Dosa, G. Generalized multifit-type methods II, \emph{Alkalmaz. Mat. Lapok}, 20(1) (2000), 91--111, 2000.
\bibitem{Dosa2001}
Dosa, G. Generalized multifit-type methods for scheduling parallel identical machines, \emph{Pure Mathematics and Applications}, 12(3), 287--296, 2001.
\bibitem{EckPinedo1993}
Eck, B. T. and M. Pinedo, On the minimization of the makespan subject to flowtime optimality, \emph {Operations Research}, 41(4), 1993.
\bibitem{Friesen1984}
Friesen, D. K., Tighter bounds for the MULTIFIT processor scheduling algorithm, \emph{SIAM Journal of Computing} 13, 179--181, 1984.
\bibitem{GareyJohnson1981}
Garey, M. R. and D. S. Johnson, Approximation algorithms for bin packing problems-a survey, in
\emph{Analysis and Design of Algorithms in Combinatorial Optimization}, G. Ausiello and M. Lucertini, (eds.), Springer-Verlag, New York, 147--172, 1981.
\bibitem{Graham1966}
Graham, R. L., Bounds for certain multiprocessing anomalies, \emph {Bell Systems Technical Journal}, 45(9), 1563--1581, 1966.
\bibitem{Graham1969}
Graham, R. L., Bounds on multiprocessing timing anomalies, \emph {SIAM Journal of Applied Mathematics}, 17(2), 416--429, 1969.
\bibitem{GuptaHo2001}
Gupta, J. N. D. and J. C. Ho, Minimizing makespan subject to minimum flowtime on two identical parallel machines, \emph{Computers and Operations Research}, 28, 705--717, 2001.
\bibitem{HoWong1995}
Ho, J. C. and J. S. Wong, Makespan minimization for $m$ parallel identical processors, \emph{Naval Research Logistics}, 42, 935--948, 1995.
\bibitem{HochbaumShmoys1987}
Hochbaum, D. S. and D. B. Shmoys, Using Dual Approximation Algorithms for Scheduling Problems: Theoretical and Practical Results, \emph{Journal of the ACM}, 34(1), 144--162, 1987.
\bibitem{Johnson1973}
Johnson, D. S., Near optimal bin packing algorithms, \emph{Ph.D. thesis}, Mathematics Dept., MIT, Cambridge, MA, 1973.
\bibitem{JohnsonDemersUllmanGareyGraham1974}
Johnson, D. S., A. Demers, J. D. Ullman, M. R. Garey and R. L. Graham, Worst-case performance bounds for simple one-dimensional packing algorithms, \emph {SIAM Journal  of Computing}, 3, 299--325, 1974.
\bibitem{LinLIoa2004}
Lin, C. H. and C. J. Liao, Makespan minimization subject to flowtime optimality on identical parallel machines,  \emph{Computers and Operations Research}, 31, 1655--1666, 2004.
\bibitem{Ravi2010}
Ravi, P. S., \emph{Techniques for Proving Approximation Ratios in Scheduling}, M.Math. Thesis, Dept. of Combinatorics and Optimization, Faculty of Mathematics, University of Waterloo,
Canada, 2010.
\bibitem{delaVegaLueker1981}
de la Vega, W. F. and G. S. Lueker, Bin packing can be solved within $1+\epsilon$ in linear time, \emph{Combinatorica}, 1(4),349--355, 1981.
\bibitem{Yue1990}
Yue, M. Y., On the exact upper bound for the multifit processor scheduling algorithm. \emph{Annals of Operations Research}, 24, 1--4, 233--259, 1990.
\bibitem{YueKellererYu1988}
Yue, M., H. Kellerer and Z. Yu, A simple proof of the inequality $R_M(MF(k)) \leq 1.2 + 1/2^k$ in multiprocessor scheduling, \emph{Report No. 124, Institut fur Mathematik, Technische Universitat Graz}, 1--10, 1988.
\end{thebibliography}
\end{document}